%% file: paper.tex
\documentclass[10pt,twocolumn,letterpaper]{article}

\usepackage{times}
\usepackage{epsfig}
\usepackage{graphicx}
\usepackage{amsmath}
\usepackage{amssymb}
\usepackage{amsthm}
\usepackage{thmtools, thm-restate}
\usepackage{wrapfig}
\usepackage[lined,boxed]{algorithm2e}
\usepackage{url} %to deal with error: undefined control sequence \url

\usepackage[utf8]{inputenc} % set input encoding (not needed with XeLaTeX)
\DeclareUnicodeCharacter{202C}{}
\usepackage{geometry} % to change the page dimensions
\usepackage{booktabs} % for much better looking tables
\usepackage{array} % for better arrays (eg matrices) in maths
\usepackage{paralist} % very flexible & customisable lists (eg. enumerate/itemize, etc.)
\usepackage{verbatim} % adds environment for commenting out blocks of text & for better verbatim
\usepackage{subfig} % make it possible to include more than one captioned figure/table in a single float
\usepackage{fancyhdr} % This should be set AFTER setting up the page geometry
\usepackage{sectsty}
\allsectionsfont{\sffamily\mdseries\upshape} % (See the fntguide.pdf for font help)
\usepackage[nottoc,notlof,notlot]{tocbibind} % Put the bibliography in the ToC
\usepackage[titles,subfigure]{tocloft} % Alter the style of the Table of Contents

\usepackage{amsthm}
\usepackage{color} %for color remarks in text
\newtheorem{theorem}{Theorem}
\newtheorem{lemma}[theorem]{Lemma}
\newtheorem{definition}{Definition}

\renewcommand{\P}{\mathcal{P}}
\newcommand{\Q}{\mathcal{Q}}

\newcommand{\C}{\mathcal{C}}
\newcommand{\RR}{\mathbb{R}}
\newcommand{\NN}{\mathbb{N}}
\newcommand{\Od}{\mathrm{O}(d)}

\newcommand{\SO}{\mathrm{SO}(d)}

\newcommand{\Y}{\mathcal{Y}}
\newcommand{\Pibi}{\Pi_{\mathrm{bi}}}
\newcommand{\PiCP}{\Pi_{\mathrm{CP}}}
\newcommand{\GCP}{G_{\mathrm{CP}}}
\newcommand{\ECP}{E_{\mathrm{CP}}}
\newcommand{\FCP}{F_{\mathrm{CP}}}
\newcommand{\Gbi}{G_{\mathrm{bi}}}
\newcommand{\Ebi}{E_{\mathrm{bi}}}
\newcommand{\Fbi}{F_{\mathrm{bi}}}
\newcommand{\EbiTilde}{\tilde{E}_{\mathrm{bi}}}
\newcommand{\ECPbar}{\bar{E}_{\mathrm{CP}}}
\newcommand{\FbiTilde}{\tilde{F}_{\mathrm{bi}}}
\newcommand{\ECPtilde}{\tilde{E}_{\mathrm{CP}}}
\newcommand{\FCPtilde}{\tilde{F}_{\mathrm{CP}}}
\newcommand{\ub}{\mathrm{ub}}
\newcommand{\lb}{\mathrm{lb}}
\newcommand{\nq}{n_{\mathrm{qBnB}}}
\newcommand{\nb}{n_{\mathrm{BnB}}}
\newcommand{\Czero}{\mathbf{\mathcal{C}}^0}
\newcommand{\Id}{I_d}
\newcommand{\J}{\mathcal{J}}
\newcommand{\eg}{e.g.}
\newcommand{\ie}{i.e.}

\begin{document}

%%%%%%%%% TITLE
\title{Linearly Converging Quasi Branch and Bound Algorithms for Global Rigid Registration}

\author{Nadav Dym\\
Duke University\\
{\tt\small nadavd@math.duke.edu‬}
\and
Shahar Ziv Kovalsky\\
Duke University\\
{\tt\small shaharko@math.duke.edu}
}

\maketitle

\begin{abstract}
In recent years, several branch-and-bound (BnB) algorithms have been proposed to globally optimize  rigid registration problems. In this paper, we suggest a general framework to improve  upon the BnB approach, which we name \emph{Quasi BnB}. Quasi BnB replaces the linear lower bounds used in BnB algorithms with quadratic quasi-lower bounds which are based on the quadratic behavior of  the energy in the vicinity of the global minimum. While quasi-lower bounds are not truly lower bounds, the Quasi-BnB algorithm is globally optimal. In fact we prove that it exhibits linear convergence -- it achieves $\epsilon$-accuracy in $~O(\log(1/\epsilon)) $ time while the time complexity of other rigid registration BnB algorithms is polynomial in $1/\epsilon $. Our experiments verify that Quasi-BnB is significantly more efficient than state-of-the-art BnB algorithms, especially for problems where high accuracy is desired.
\end{abstract}

%%%%%%%%% BODY TEXT
\section{Introduction}\label{sec:intro}

\begin{figure}[t]
	\vspace{-10pt}
	\includegraphics[width=\columnwidth]{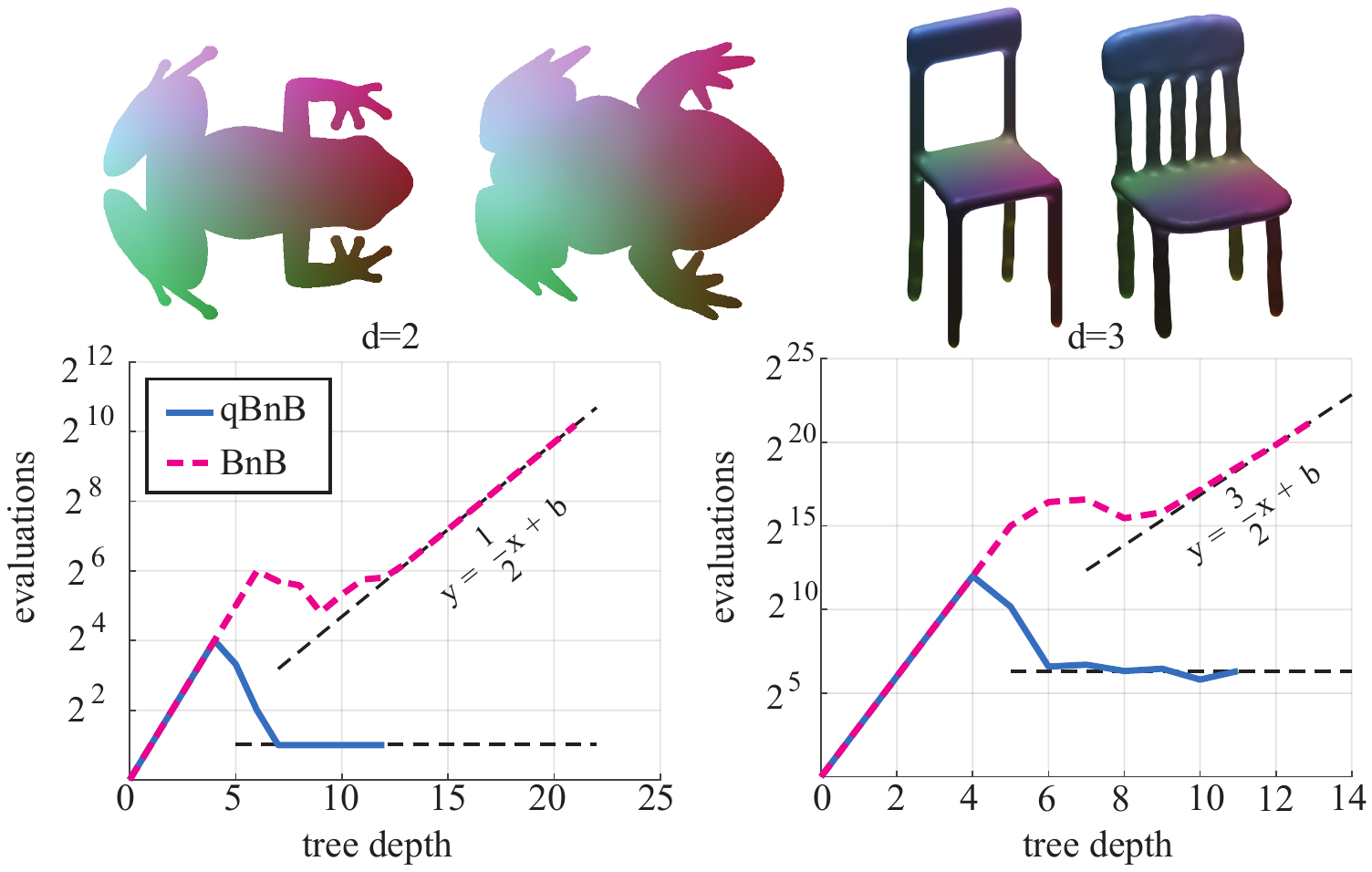}
		\caption{Comparing the proposed qBnB algorithm and the BnB algorithm of \cite{li20073d} for rigid bijective problems. 
			The colored shapes illustrate the alignment and correspondence maps obtained for a pair of 2D and 3D shapes. The graphs demonstrate the 
			preferable complexity of our qBnB approach, dashed lines illustrating the asymptotic behavior derived in our analysis (Theorem~\ref{thm:complexity}).} 
	\label{fig:hartley} \vspace{-10pt}
\end{figure}

Rigid registration is a fundamental problem in computer vision and related fields. The input to this problem are two similar shapes and the goal is to find the rigid motion that best aligns the two shapes, as well as a good mapping between the aligned shapes. There are several different formulations of the rigid registration problem. In this paper we focus on two popular formulations, the rigid closest point \emph{(rigid-CP)} \cite{besl1992method}  problem which is appropriate for partial matching problems (\eg where a partial point cloud obtained from a scanning device is mapped to a reconstructed model), and the   \emph{rigid-bijective} problem \cite{rangarajan1997softassign} which is applicable to full matching problems where it naturally defines a distance between shapes \cite{al2013continuous}.

Both of these rigid registration formulations optimize an appropriate energy $E(g,\pi) $ that depends on the chosen rigid motion $g$ and a correspondence $\pi$ between the two given point clouds. This energy is typically non-convex, however it is \emph{conditionally tractable} in that there exist efficient algorithms for globally optimizing for either of its variables when the other is fixed. This is put to good use by the well-known \emph{iterative closest point} (ICP) algorithm \cite{besl1992method} that optimizes $E$ using a very efficient alternating approach. It converges, however, to a local minimum and thus it strongly depends on a good  initialization.
%It convergence is however to a local minimum and so its success depends on the availability of a good initialization.

Conditional tractability can also be useful for \emph{global} optimization algorithms. For a fixed rigid motion $g$ it is possible to compute
$$F(g)=\min_\pi E(g,\pi)$$ 
and so to optimize $E$ globally it suffices to optimize $F$. Figure~\ref{fig:illust}(a)  shows an example function $F$ derived from a 2D rigid registration problem. The function $F$ is non-convex and non-differentiable,  
and the complexity of global optimization algorithms for optimizing $F$ is typically exponential in the dimension $D$ of the rigid motion space. However, in many cases of interest $D$ is a small constant and so these algorithms are in fact tractable. Thus, while high dimensional rigid registration problems can be computationally hard \cite{dym2017exact}, they are fixed parameter tractable~\cite{cygan2015parameterized}.

Branch and bound (BnB) algorithms \cite{bustos2016fast,yang2016goicp,pfeuffer2012discrete,li20073d}  are perhaps the only algorithms with a deterministic guarantee to globally optimize the rigid registration problem. The key ingredient in BnB algorithms is the ability to give a lower bound for the value of $F$ in a certain region, based on a single evaluation of $F$. For example, a lower bound for the grey interval in Figure~\ref{fig:illust}(c) would attempt to estimate the minimum of $F$ on the interval (denoted by the purple line (d)), based on the blue point in the center of the interval. 

\begin{figure}[t]
	\centering
	\includegraphics[width=0.8\columnwidth]{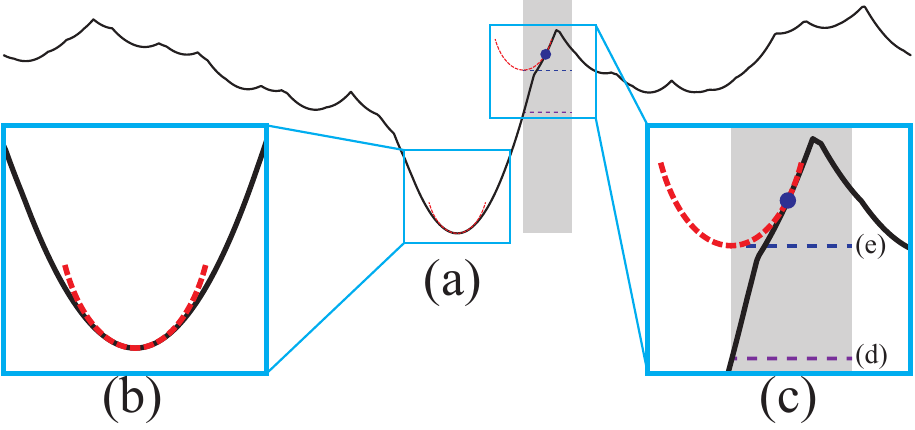}
	\caption{(a) Illustration of the function $\FCP$ defined in Subsection~\ref{sub:computing} for a 2D registration problem as a function of rotation angle. (b) Quadratic behavior of $\FCP$ at a minimum. (c)  The true minimum (d) on the grey interval and a quasi-lower bound (e) computed using the quadratic behavior near the global minimum.} 
	\label{fig:illust}
\end{figure}

 In this paper we introduce the term \emph{quasi-lower bound}, which is  a lower bound for the value of the function, over a given region, under the \emph{possibly-false} assumption that the region contains a global minimum. While quasi-lower bounds are not true bounds, the global optimality of qBnB algorithms is not compromised if lower bounds are replaced by quasi-lower bounds - leading to algorithms that we name quasi-BnB \emph{(qBnB)}. 

The advantage of quasi-lower bounds is that for conditionally smooth energies (as defined in Definition~\ref{def}) $F$ exhibits quadratic behavior near the global minimum, and this can be leveraged to define quadratic quasi-lower bounds, in contrast with lower bounds that are typically linear. This leads to much improved bounds, as illustrated in Figure~\ref{fig:illust}, where the quasi-lower bound (e) is able to provide a  bound that is higher than the true minimal value (d) in the grey interval, and hence tighter than any lower bound. 

The tightness of quasi-lower bounds yields substantial benefits in terms of complexity --we show that the complexity of achieving an $\epsilon$-optimal solution using BnB is at best $\sim \epsilon^{-D/2}$. In contrast, under some weak assumptions, the complexity of the qBnB algorithm  is $\sim \log(1/\epsilon)$, that is, it achieves a \emph{linear convergence rate} (in the sense of \cite{boyd2004convex}). Our experiments (\eg Figure~\ref{fig:hartley}) verify these theoretical results, and show that qBnB can be considerably more efficient than state-of-the-art BnB algorithms, especially in the presence of high noise or when high accuracy is required.

\section{Related Work}
We focus on global optimization methods for rigid registration. For other aspects of the rigid registration problem we refer the reader to surveys such as  \cite{tam2013registration}.

A variety of methods from the global optimization toolbox have been used to address the rigid registration problem, including
Monte Carlo \cite{irani1999combinatorial}, RANSAC based \cite{aiger20084}, particle filtering \cite{sandhu2010point}, particle swarming \cite{wachowiak2004approach}, convex relaxations \cite{maron2016point,khoo2016non}, and graduated non-convexity \cite{zhou2016fast}. Unlike BnB algorithms, none of these algorithms come with a deterministic guarantee for global convergence.

In \cite{pfeuffer2012discrete,mount1999efficient,breuel2003implementation} BnB algorithms for rigid matching in 2D are suggested. In \cite{gelfand2005robust} an algorithm for rigid matching in 3D is proposed, based on a combination of matching rotation-invariant descriptors and BnB. In \cite{li20073d} a BnB algorithm for a modified bijective matching objective in 3D is suggested, based on Lipschitz optimization. 

In \cite{yang2016goicp,yang2013goicp} the Go-ICP algorithm for rigid matching in 3D is suggested. Their algorithm searches the 6D transformation space very efficiently in the low noise regime, although it can be slow in the high noise regime. A branch and bound procedure for  the rigid matching  energy of \cite{breuel2003implementation} in 3D is proposed in \cite{bustos2016fast,parra2014fast}. For this energy they suggest various speedups that are not available for the classic maximum likelihood objective considered in ICP. 

The quasi-BnB approach we propose here can boost the convergence of any of the methods above  which consider  conditionally smooth objective, such as \cite{yang2016goicp,pfeuffer2012discrete,li20073d}, by simply replacing their first order lower bound with an appropriate second order quasi-lower bound. However it is not directly applicable for optimizing objectives which are not conditionally smooth such as the geometric matching objective of \cite{bustos2016fast,breuel2003implementation}.

 BnB algorithms with quadratic lower bounds have been proposed for smooth unconstrained optimization \cite{fowkes2013branch}, and   for camera pose estimation \cite{hartley2009global}. The latter approach obtains quadratic lower bounds by jointly optimizing over a linear approximation of a rotation, and the remaining variables; This approach does not straightforwardly apply to our scenario. In contrast, our qBnB approach only requires solving optimization problems with a fixed rotation.
 
\cite{pottmann2006geometry} show that ICP achieves linear converge rate, and suggest alternative second order methods with super-linear convergence rate. These algorithms only converge to a local minimum, whereas our algorithm is shown to exhibit linear convergence to the \emph{global minimum}.

\section{Method}
\subsection{Problem statement}\label{sub:problem}

%\sk{alternative suggestion draft: Rigid registration takes many different forms and variants. We first explain the common structure of problems for which our quasi-BnB framework is applicable, which we will name in this paper \emph{$D$-quasi optimizable}. Then, we focus on two instances of the rigid registration problem.}
We consider the application of the quasi BnB framework as described in Section~\ref{sec:intro}
 to the problem of rigid alignment of point clouds. This problem has several variants, and in this paper we will focus on two of them that we will refer to as the rigid-closest-point (\emph{rigid-CP}) problem and the \emph{rigid-bijective} problem. We also wish to focus on the common structure of problems for which our quasi BnB framework is applicable, which we will name in this paper \emph{$D$-quasi-optimizable}. We begin by defining $D$-quasi-optimizable problems.
 
 Consider optimization problems of the form 
        \begin{equation}\label{e:generalEnergy}
 \min_{x \in \RR^D, y \in \Y} E(x,y),
 \end{equation}
 where $\Y$ is a finite set. Denote the cube centered at $x$ with half-edge length $h$ by $\C_h(x)$. 
 \begin{definition}\label{def}
        We say that an optimization problem of the form \eqref{e:generalEnergy} is $D$-quasi-optimizable in a cube $\Czero=\C_{h_0}(x_0) $ if it satisfies the following conditions:
        \begin{enumerate}
                \item \emph{Existence of a minimizer:} There exists a minimizer $(x_*,y_*) $ of $E$ such that $x_*\in\Czero $. 
                \item \emph{$D$-tractability}: For fixed $x $, minimizing $E(x,\cdot)$ over $\Y$ can be performed in polynomial time.
                \item \emph{Conditional smoothness} For fixed $y \in \Y $, the function $E(\cdot,y)$ is smooth.
        \end{enumerate}  
        \end{definition}

We now turn to defining the rigid-CP and rigid-bijective problems and explaining why they are $D$-quasi-optimizable.  

 The input to the rigid-CP problem are two point clouds $\P=\{p_1,\ldots, p_n\}$ and $\Q=\{q_1,\ldots ,q_m \}$ in $\RR^d$ where we assume $d=3$ or $d=2$. We assume the point clouds are normalized to be in  $[-1,1]^d$ and have zero mean. Our goal is to find a rigid motion that aligns the points as-well-as-possible. Namely, denoting the group of (orientation preserving) rigid motions $\SO \times \RR^d $ by $\GCP$, and denoting by  $\PiCP$ the collection of all functions $\pi:\P \to \Q$, our goal is to solve the minimization problem
%\begin{equation}\label{eq:rigid}
%\hspace{-2em} \min_{\hspace{2em}(R,t) \in \GCP, \pi \in \PiCP} \hspace{-3em} %\ECP(R,t,\pi)= \frac{1}{n} \sum_i \|Rp_i+t-q_{\pi(i)}\|^2    
%\end{equation}
\begin{equation}\label{eq:rigid}
\hspace{-2em} \min_{\substack{(R,t) \in \GCP \\ \pi \in \PiCP}} \hspace{-1em} \ECP(R,t,\pi)= \frac{1}{n} \sum_{i=1}^n \|Rp_i+t-q_{\pi(i)}\|^2    
\end{equation}
Following \cite{yang2016goicp,li20073d} we simplify the domain of \eqref{eq:rigid}  by using the exponential map: We use $s$ to denote the intrinsic dimension of $\SO$. For a vector $r \in \RR^s$, we define $[r]$ to be the unique $d \times d $ skew-symmetric matrix ($[r]^T$=$-[r]$) whose entries under the diagonal are given by $r$. By applying the matrix exponential to $[r]$ we get a matrix 
$$R_r=\exp([r]) \in \SO .$$
Every rotation can be represented as $R_r$ for some $r$ in the closed ball $B_{\pi}(0)$ 
and so \eqref{eq:rigid} can be reduced to the problem 
\begin{equation} \label{e:CPequiv}
\min_{(r,t) \in \RR^s \times \RR^d, \pi \in \PiCP} \ECP(R_r,t,\pi).\end{equation}
Our next step is to identify a cube $\Czero$ in $\RR^s \times \RR^d$ in which $\ECP $ must have a global minimum. For the rotation component we can take the  cube $\C_\pi(0) $ that bounds the ball $B_{\pi}(0)$ .  For the translation component, we note that if  $(R_*,\pi_*,t_*)$ is a minimizer of $E$, then the optimal translation is the difference between the average of $R_* p_i $, which is zero by assumption, and the average of $q_{\pi(i)}$, which will be in the unit cube. Thus there exists a minimizer $(r_*,t_*,\pi_*) $ such that $(r_*,t_*) $ is in $\C_{\pi}(0) \times \C_1(0) $. This shows that \eqref{e:CPequiv} satisfies the first condition for being $D$-quasi-optimizable, where ${x=(r,t)}$, $y=\pi $ and $D=s+d$ (\ie, $D=6 $ or $3$ for 3D and 2D problems respectively). 

$D$-tractability follows from the fact that the optimal mappings $\pi$ minimizing $\ECP(R,t,\cdot)$ for fixed $R,t $ are  just the mappings that take each rigidly transformed point $Rp_i+t$ to its closest point in $\Q$, hence the term rigid-CP. Conditional smoothness is obvious. 

  The rigid bijective problem is similar to rigid-CP, but focuses on the case $n=m$, and only allows mappings between $\P$ and $\Q$ that are bijective (\ie, permutations).   We denote this set of mappings by $\Pibi$. In this scenario optimizing for $\pi$ while holding the rigid transformation component fixed is tractable, but more computationally intensive than rigid-CP since it requires solving a linear assignment problem. On the other hand, in this scenario the optimal translation is always $t_*=0$, since the mean of  the points $p_i $ and $q_{\pi(i)}$ is zero. Therefore we can reduce our problem to lower-dimensional optimization over $\Gbi=\SO $. As before, we use the exponential map to reparameterize the problem as
 \begin{equation}
 \hspace{-10pt}\min_{r \in \RR^s, \pi \in \Pibi} \hspace{-12pt}\Ebi(R_r,\pi)=\frac{1}{n} \sum_i \|R_{r}p_i-q_{\pi(i)}\|^2
 \end{equation}

% \begin{wraptable}[6]{r}{0.4\columnwidth}
% 	\begin{tabular}{|l|l|l|}
% 		\hline
% 		$D$       & 2D  &  3D\\
% 		\hline
% 		CP        &  3      &  6   \\
% 		\hline
% 		bijective &  1      &  3\\
% 		\hline
% 	\end{tabular}
% \vspace{-0.2cm}
% 	\caption{}
% 	 	\label{tab:D}
% \end{wraptable}
It follows from our discussion above that this optimization problem is $D$-quasi-optimizable over $\C_\pi(0) $ with $x=r,y=\pi$ and $D=s $ (\ie, $D=3 $ or $1$ for 3D and 2D problems respectively). 

\subsection{Optimizing quasi-optimizable functions}\label{sub:optimizing}
\paragraph{BnB algorithms.} $D$-tractability implies that we can reduce \eqref{e:generalEnergy} to the equivalent problem of minimizing the $D$-dimensional function
\begin{equation}\label{eq:F}
F(x)=\min_{y \in \Y} E(x,y). \end{equation} 
BnB algorithms for rigid registration and related problems are typically based on the ability to show that for $\delta>0 $ and $x_1,x_2 \in \Czero $ satisfying $\|x_1-x_2\|<\delta $,
\begin{equation}\label{eq:Delta}
F(x_1)-F(x_2)\leq \Delta(\delta)=L\delta+O(\delta^2) .
\end{equation}
Using a bound of this form
$F$  can be bounded from below \emph{in a cube} $\C_{h_i}(x_i) $ by evaluating $F(x_i)$ and noting that for any $x$ in the cube, \eqref{eq:Delta} implies that $F(x)$ is larger than 
\begin{equation}\label{eq:lb}
\lb_i\equiv F(x_i)-\Delta(\sqrt{D}h_i) . 
\end{equation}
In the appendix we show that any $F$ arising from a $D$-quasi-optimizable optimization problem is Lipschitz and thus a bound of the form \eqref{eq:Delta} always exists.

Based on the lower bound \eqref{eq:lb}, a simple breadth-first-search (BFS) BnB algorithm starts from a coarse partitioning of $\Czero$ into sub-cubes $\C_{h_i}(x_i)$, and then evaluates $F$ at each of the $x_i$. Each such evaluation gives a global upper bound $\ub_i=F(x_i) $ for the minimum $F^*$, and a local lower bound $\lb_i$ for the minimum \emph{on the cube} as defined in \eqref{eq:lb}. At each step the algorithm keeps track of the best global upper bound found so far, which we denote by $\ub$, and on the global lower bound $\lb$ that is defined as the smallest local lower bound found in this partition. Now for every $i$, If $\lb_i>\ub$ then $F$ is not minimized in $\C_{h_i}(x_i) $ and this cube can be excluded from the search. The BnB algorithm then refines the partition into smaller cubes for cubes that have not yet been eliminated, and this process is continued until $\ub-\lb<\epsilon $. This is the BnB strategy used in \cite{yang2016goicp,pfeuffer2012discrete,li20073d,mount1999efficient,breuel2003implementation} for rigid registration problems, though these papers vary in the search strategy they employ, the rigid energy the consider, the bound \eqref{eq:Delta} they compute, and other aspects.

\paragraph{Quasi-BnB.}
% The qBnB algorithm we suggest in this paper is based on another "smooth-like" property of $F$: like smooth functions, the rate of change of $F$ around a global minimizer is quadratic:
 The qBnB algorithm we suggest in this paper is based on the observation that the non-differentiable $F$ resembles smooth functions in that its behavior near minimizers is quadratic:
 \begin{lemma}
 \item If $x_*\in \Czero$ is a global minimizer of $F$, then there exists $C>0$ such that 
 \begin{equation}\label{eq:quadratic}
 F(x)-F(x_*) \leq C\|x-x_*\|^2, \forall x \in \Czero
 \end{equation} 
 \end{lemma}
%\emph{Proof.} %
\begin{proof}
Assume $(x_*,y_*) $ is a global minimizer of $E$, Let $N_{x,y}$ be the operator norm of the Hessian of the function $E(\cdot,y) $ at a point $x$, and let $C$ be  the maximum of $N_{x,y}$ over $\Czero \times \Y $. Since the first order approximation of $E(\cdot,y_*) $ at $x_*$ vanishes, we obtain
\begin{align}\label{eq:quadBound}
F(x)-F(x_*)&\leq E(x,y_*)-E(x_*,y_*)\\
           &\leq C\|x-x_*\|^2  \qedhere \nonumber 
\end{align}
\end{proof} 
Let us  assume that we have an explicit quadratic bound of the general form \eqref{eq:quadBound}. That is, if $x_*$ is a global minimizer, and  $x,x_* \in \Czero $ satisfy $\|x-x_*\|<\delta $, then
\begin{equation}\label{eq:DeltaStar}
F(x)-F(x_*)\leq \Delta_*(\delta)=C\delta^2+O(\delta^3) .
\end{equation}
In Subsection~\ref{sub:computing} we give an explicit computation of $\Delta_*$ for rigid-bijective and rigid-CP problems. While $\Delta_*$ only induces a bound on $F$ near a global minimum, it can be used similarly to the Lipschitz bounds to eliminate cubes $\C_{h_i}(x_i)$, for if the global minimizer $x_*$ is contained in $\C_{h_i}(x_i) $ then $F(x_*)$ cannot be larger than
$$ \lb_i \equiv F(x_i)-\Delta_*(\sqrt{D} h_i). $$
Thus, $\C_{h_i}(x_i)$ can be eliminated whenever ${\lb_i>\ub}$. 
We emphasize that $\lb_i$  is not a lower bound for the value of $F$ in the cube since if the cube does not contain a global minimizer $F$ may vary linearly inside the cube -- for this reason we name it a \emph{quasi-lower bound}. Nonetheless from the point of view of BnB algorithms (see Figure~\ref{fig:illust}), linear lower bounds can be replaced by quadratic quasi-lower bounds, leading to a significantly more efficient qBnB algorithm, still guaranteed to converge to an $\epsilon$-optimal solution. A simple BFS qBnB algorithm for optimizing $D$-quasi-optimizable functions, for a given quasi-lower bound $\Delta_*(\delta) $, is provided in Algorithm~\ref{alg}. 

\begin{algorithm}[t]
\SetKwInOut{Input}{input}\SetKwInOut{Output}{output}
        \SetAlgoLined
        \Input{Required accuracy $\epsilon$} 
        \Output{$\epsilon$-optimal solution $x_*$}
        $\ub \leftarrow \infty$, $\lb \leftarrow -\infty$, $g \leftarrow 0$\;
    Put $\Czero=\C_{h_0}(x_0) $ into the list $L_g$ \;
    \While{$\ub-\lb>\epsilon$}{
        %\For{$i\leftarrow 1$ \KwTo $\mathrm{Length}(L_g)$}{
        \ForAll{$\C_{h_i}(x_i) \in L_g$}{
                        Compute $F(x_i) $\;
                        $\ub_i \leftarrow F(x_i) $, $\lb_i \leftarrow F(x_i)-\Delta_*(\sqrt{D}h_i) $\;
                        \If{$\ub_i<\ub$}{$\ub \leftarrow \ub_i$\;
                                          $x_* \leftarrow x_i $\;}
                }
         $\lb=\min_i \lb_i $ \;
                %\For{$i\leftarrow 1$ \KwTo $\mathrm{Length}(L_g)$}{
                \ForAll{$\C_{h_i}(x_i) \in L_g$}{
                        \If{$\lb_i\leq \ub$ }{subdivide $\C_{h_i}(x_i)$ into $2^D$ sub-cubes with half-edge length $h_i/2$ and insert into $L_{g+1}$}
                        %\If{$\lb_i\leq \ub$ }{subdivide $\C_{h_i}(x_i)$ into $2^D$ sub-cubes and insert into $L_{g+1}$}
                }
         $ g \leftarrow g+1 $\;
        }
\caption{BFS qBnB} 
        \label{alg}
\end{algorithm} 

\subsection{Complexity} Theorem~\ref{thm:complexity} (below, proof in appendix)  provides complexity bounds for the qBnB algorithm with a quadratic $\Delta_*(\delta)$, as in Algorithm~\ref{alg}, and for a similar BnB algorithm, where $\Delta_*$ is replaced by a Lipschitz bound $\Delta(\delta)$.
 We denote the number of $F$-evaluations in the algorithms by $\nq $ and $\nb$ respectively. We consider the limit where the prescribed accuracy $\epsilon$ tends to zero and all other parameters of the problem are held fixed:
\begin{restatable}{theorem}{complexity}\label{thm:complexity}
There exist positive constants $C_1,\ldots,C_4 $, such that 

 \begin{equation}C_1\epsilon^{-D/2} \leq  \nb \leq C_2 \epsilon^{-D} . \end{equation}
 \begin{equation} \nq \leq C_3 \epsilon^{-D/2}.  \end{equation}
 Furthermore, if $E $ has a finite number of minimizers $(x_{\ell}^*,y_{\ell}^*)_{\ell=1}^N $, 
 and the Hessian of $E(\cdot,y_{\ell}^*) $ is strictly positive definite for all $\ell$, then 
        \begin{equation} \label{eq:linearConvergence}
        \nq \leq C_4 \log_2(1/\epsilon) .
        \end{equation}

\end{restatable}

\subsection{Computing quasi-lower bounds}\label{sub:computing}
To conclude our discussion we now need to state explicit quadratic quasi-lower bounds for the rigid-CP and rigid-bijective problems. We denote by $\FCP$ and $\Fbi$ the functions obtained by applying \eqref{eq:F} to $\ECP$ and $\Ebi$.  

\paragraph{Rigid bijective quasi-lower bounds.} For $k \in \NN$, let $\psi_k$ denote the truncated Taylor expansion of $e^x$,
\begin{equation}
\psi_k(x)=e^x-\sum_{j=0}^{k-1} \frac{x^j}{j!} 
\end{equation}
%\vspace{-10pt}
 \begin{restatable}{theorem}{qlbBij}\label{thm:quasiBi}
        Let $\delta>0, r \in \RR^D $ and $r_*$ be a global minimizer of $\Fbi$, and assume $\|r-r_*\| \leq \delta $. Let $\sigma_{\P},\sigma_{\Q} $ denote the Frobenius norm of the matrices whose columns are the points in $\P$ and $\Q$ respectively. Then $\Delta_*(\delta)$ is given by                 
        \begin{equation}\label{eq:bijQlb}
        \hspace{-7pt}\Fbi(r)-\Fbi(r_*)\leq \Delta_*(\delta)\equiv \frac{2}{n} \sigma_{\P} \sigma_{Q} \, \psi_2(\delta)
        \end{equation}
        \end{restatable}
 \begin{proof}
Let us consider the case where   $r_*=0$, so $(\Id,\pi_*) $  minimizes $\Ebi$ for some appropriate $\pi_*$. Note that for all $R,\pi $, 
\begin{equation}
%\Ebi(R,\pi)=\frac{1}{n}\left[\sigma_{\P}^2+\sigma_{\Q}^2-2\sum_i \langle Rp_i,q_{\pi(i)} \rangle \right] 
\hspace{-10pt}\Ebi(R,\pi)=\frac{\sigma_{\P}^2+\sigma_{\Q}^2-2\sum_i \langle Rp_i,q_{\pi(i)} \rangle}{n}
\end{equation}
it follows that
\begin{align*}
\Fbi(r)&-\Fbi(0) \leq \Ebi(R_r,\pi_*)-\Ebi(\Id,\pi_*)\\
          &\stackrel{\text(*)}{=}-\frac{2}{n} \sum_{i=1}^n \sum_{k=2}^{\infty} \langle \frac{1}{k!} [r]^k p_i,q_{\pi(i)} \rangle \\
          & \leq \frac{2}{n} \sum_{i=1}^n \sum_{k=2}^\infty \frac{1}{k!}\|[r]\|_{\mathrm{op}}^k\|p_i\|\|q_{\pi_*(i)}\| \\
          &\leq \frac{2}{n}\psi_2(\|[r]\|_{\mathrm{op}})\sigma_{\P}\sigma_{\Q}.    
\end{align*}
To obtain the equality $(*)$ note that $r_*=0$ minimizes $\Ebi(\cdot,\pi_*) $ and so the first order approximation of this function vanishes. In the appendix we show that $\|[r]\|_{\mathrm{op}} \leq \|r\| $ which concludes the proof for the case $r_*=0$, and then use this case to prove the theorem for general $r_*$.
 \end{proof}
Based on this quasi-lower bound, our algorithm for optimizing the rigid-bijective problem is just Algorithm~\ref{alg}, where as stated in Subsection~\ref{sub:problem} we take $x=r,y=\pi, \C_{h_0}(x_0)=\C_{\pi}(0) $, and we use   $\Delta_*$  defined in \eqref{eq:bijQlb}.
 
\paragraph{Rigid CP quasi-lower bounds.} For the rigid-CP problem we obtain the following quasi-lower bound that we prove in the appendix. 
\begin{restatable}{theorem}{qlbCP}\label{thm:quasiICP}
        Let $(r_*,t_*) $ be a minimizer of $\FCP$, and let $(r,t) \in \RR^s \times \RR^d $, and $\delta_1,\delta_2>0$ which satisfy
        $\|r-r_*\| \leq \delta_1 $and$\|t-t_*\| \leq \delta_2$. 
        Let $f_*$  be some upper bound for the global minimum of $\FCP$.
        Then
        \begin{equation}
        \FCP(r,t)-\FCP(r_*,t_*)\leq \Delta_*(\delta_1,\delta_2)
        \end{equation}
        where
\begin{align}\label{eq:quasiLBicp}
 \Delta_*(\delta_1,\delta_2) &=\frac{1}{n}\bigg[2\psi_2(\delta_1) (\sigma_{\P}^2+\sigma_{\P}\sqrt{n f_*}) \nonumber \\
&+2\delta_2\psi_1(\delta_1)\sum_i \|p_i\|+n\delta_2^2   \bigg] 
\end{align}
\end{restatable}
%\vbox{
Using this quasi-lower bound the most straightforward way to construct a qBnB algorithm is by using Algorithm~\ref{alg} setting $x=(r,t), y= \pi $ and using the initial cube $\C_{\pi}(0) \times \C_1(0)$ as  described in Subsection~\ref{sub:problem} and the bound $\Delta_*(\delta_1,\delta_2) $ from \eqref{eq:quasiLBicp}. However, to enable simple comparison with the Go-ICP algorithm \cite{yang2016goicp} we use their BnB architecture, wherein two nested BnB are used -- an outer BnB for the rotation component and a separate inner BnB for the translation component. The quadratic quasi-lower bounds for these BnBs can be computed by setting $\delta_2=0 $ or $\delta_1=0$, respectively. Further details are provided in the appendix.
%}

\begin{figure}[t]
	\begin{tabular}{@{}c@{}c}
		\includegraphics[width=0.5\columnwidth]{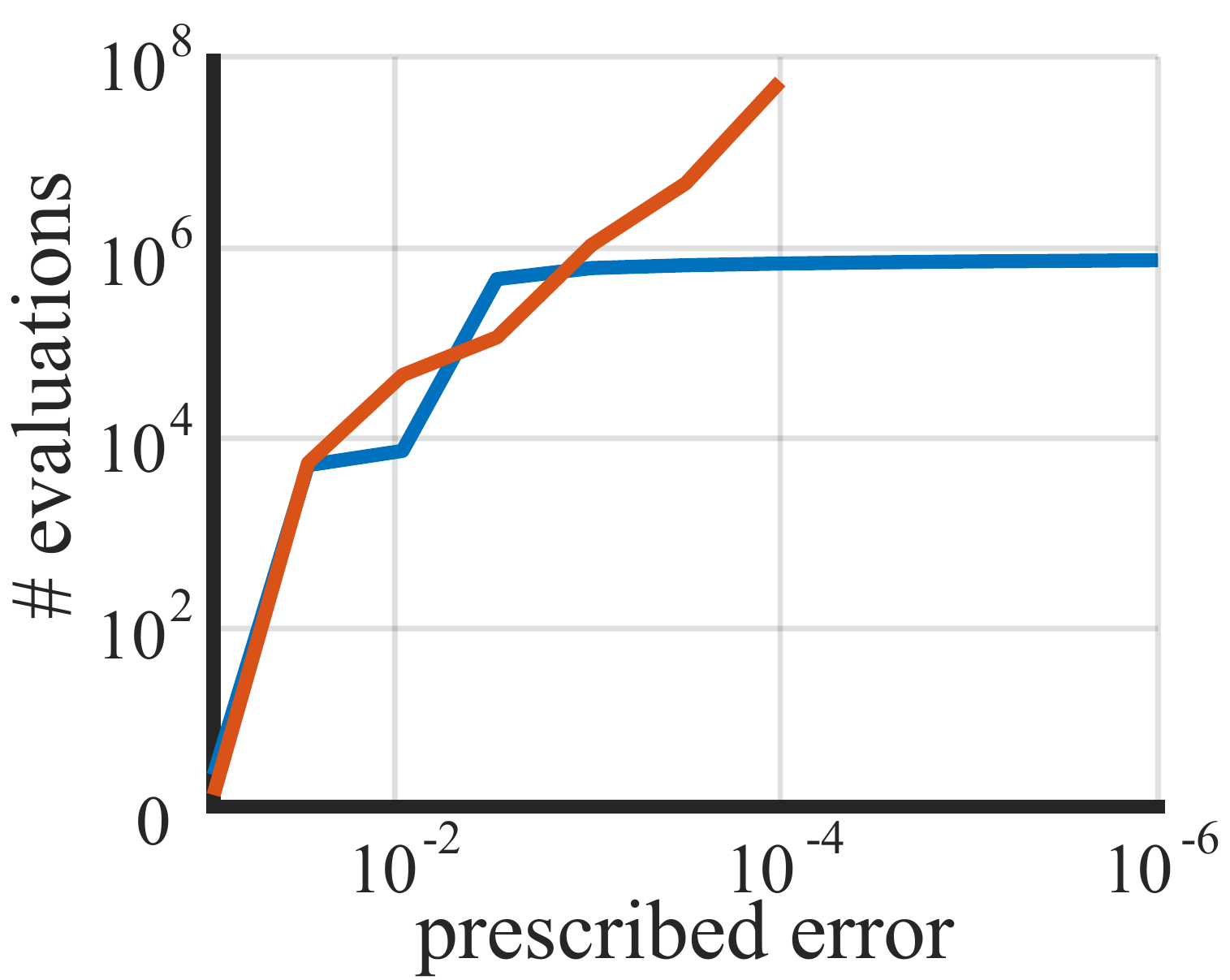}&
		\includegraphics[width=0.5\columnwidth]{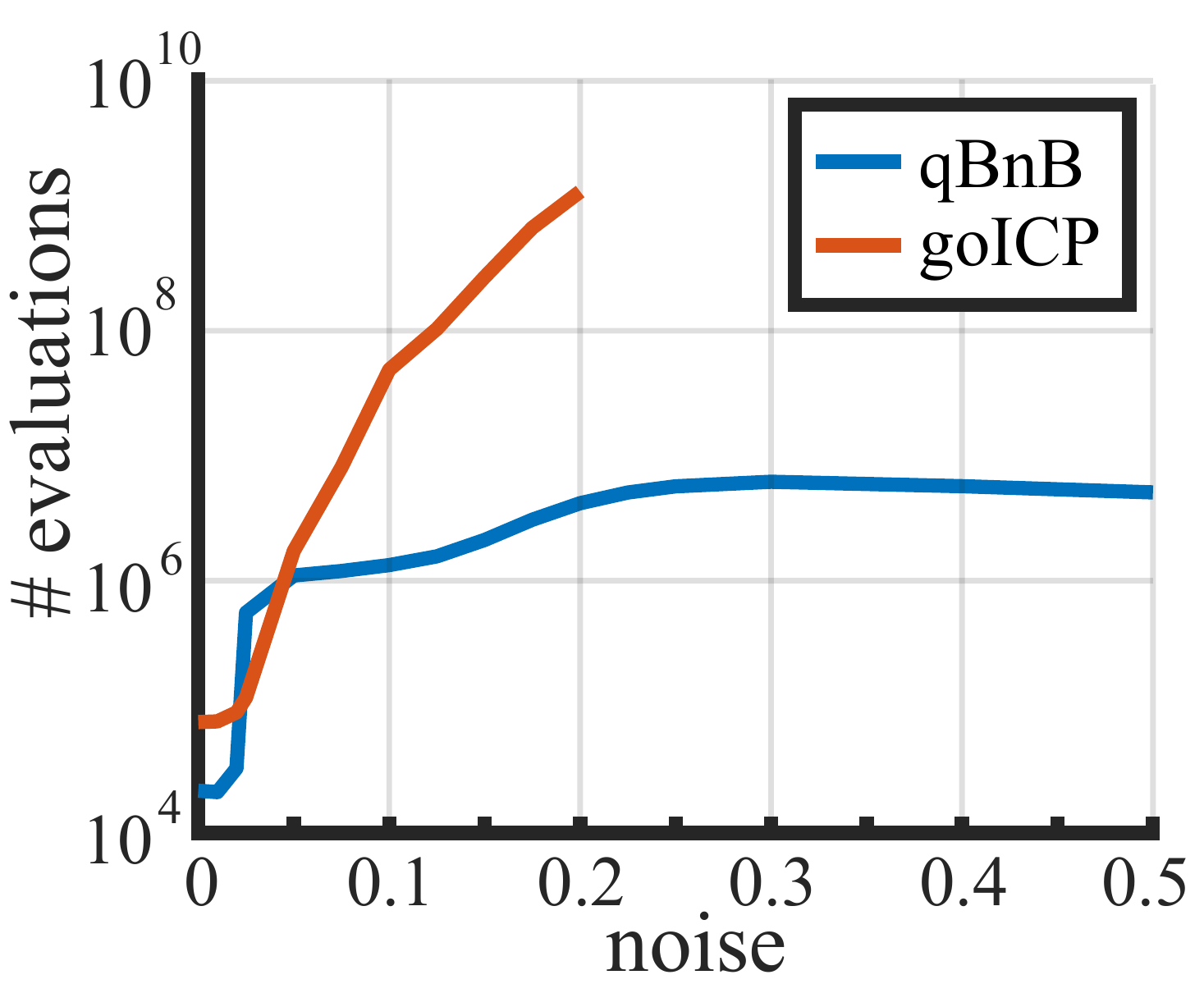}\\
		\ \ \ \ (a) & \ \ \ \ \ \ \ \ (b) 
	\end{tabular}
	\vspace{-10pt}
	\caption{Comparison of the dependence of qBnB and Go-ICP on the error tolerance (a) and noise level (b). } 
	\label{fig:goICP}
	\vspace{-10pt}
\end{figure}

\section{Results}

\paragraph{Evaluation of qBnB for rigid-CP.}

%\begin{figure*}[t]
%        \begin{tabular}{@{}c@{}c@{}c}
%                \includegraphics[width=0.3\textwidth]{evalsVsError}&
%                \includegraphics[width=0.3\textwidth]{evalsVsNoise}&
%                \includegraphics[width=0.3\textwidth]{timeVsNoise}\\
%                \ \ \ \ (a) & \ \ \ \ \ \ \ \ (b) & \ \ \ \ \ \ \ \  (c)
%        \end{tabular}
%        \caption{Comparison of the dependence of qBnB and goICP on the error tolerance (a) and noise level (b) and (c). } 
%        \label{fig:goICP}
%\end{figure*}    

We evaluate the performance of the proposed qBnB algorithm for the rigid-CP problem. We compare with Go-ICP \cite{yang2016goicp}, a state-of-the-art approach for global minimization of this problem. Our implementation is based on a modified version of the Go-ICP code, which utilizes a nested BnB (see section~\ref{sub:computing}) and an efficient approximate CP computation (see implementation details below).

We ran both algorithms on synthetic rigid-CP problems that were generated by uniformly sampling $n$ points on the unit cube to form $\P$, and then applying a random rigid transformation and Gaussian noise with std $\sigma$ to form $\Q$. Finally $\P$ and $\Q$ are translated and scaled to have zero mean and reside in the unit cube.

 Figure~\ref{fig:goICP} shows comparisons performed with varying (a) prescribed accuracy and (b) noise. Both algorithms are comparable in the low noise/low accuracy regime. This is consistent with \cite{yang2016goicp} that report better performance in the regime where the optimal energy $E^*$ is lower than the error tolerance $\epsilon$. In contrast, when accuracy or noise are increased, the complexity of the Go-ICP algorithm grows rapidly, while the complexity of qBnB stabilizes at an almost constant number of $\small{\approx} 10^6$ function evaluations. This results is consistent with the complexity analysis of Theorem~\ref{thm:complexity}.
%This too is consistent with the complexity comparison of the qBnB and  BnB algorithms presented in Theorem~\ref{thm:complexity}.
               
In Figure~\ref{fig:goICP}(a) we took $n=50$, $\sigma=0.05$, and prescribed accuracy $\epsilon$ varying between $10^{-1}$ and $10^{-6}$, and in Figure~\ref{fig:goICP}(b) we took $n=100$, $\epsilon=10^{-3}$ and $\sigma$ varying between $0$ and $0.5$. Each data point in the graph represents an average over 100 instances with the same parameters. We exclude results for Go-ICP with $\epsilon<10^{-4}$ or $\sigma>0.2 $ due to excessive run-time (over six hours). The maximal run-time of our algorithm over all experiments presented was $9$ minutes.

Figure~\ref{fig:bunny} exemplifies the applicability of our algorithm for solving larger rigid-CP problems. We used our algorithm and Go-ICP to register $500$ points sampled from the Stanford bunny \cite{standfordBunny} to a full 3D model consisting of nearly 36K vertices.
Both methods are global, thus they obtained the same results up to small indiscernible differences; therefore, we only show our alignments, as well as the evaluation counts of both algorithms. As in Figure~\ref{fig:goICP}(b), qBnB is more efficient than Go-ICP, especially in the high noise regime. For noise level of $\sigma=0.1$ our algorithm required under $7$ minutes, whereas Go-ICP took over $6$ hours to complete.  

\begin{figure}   	
		\centering
        \includegraphics[width=.95\columnwidth]{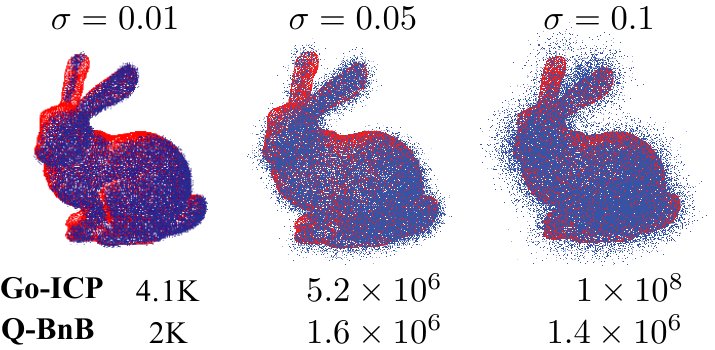}
        \vspace{-5pt}
        \caption{Registration of a noisy scan of the Stanford bunny to the full 3D model at three different noise levels. The bottom row shows the number of $\FCP $ evaluations Go-ICP and qBnB required to perform the registrations. } 
        \label{fig:bunny}
        \vspace{-10pt}
\end{figure}

\begin{figure*}[t]
	\includegraphics[width=\textwidth]{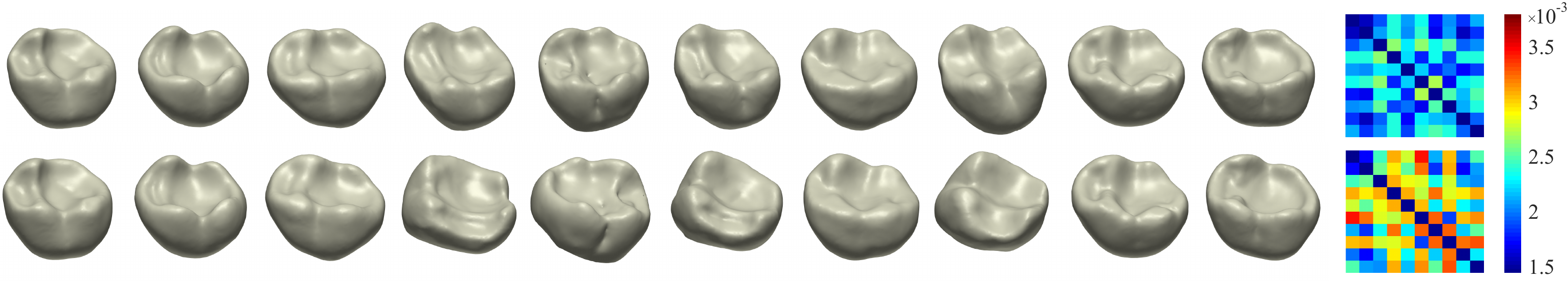}
	\vspace{-15pt}
	\caption{The top row shows the alignment computed by qBnB between the first tooth on the left and all remaining teeth. The matrix on the right shows all pairwise distances computed for the teeth. The bottom row shows the analogous result for Auto3dgm \protect{\cite{boyer2015new}}, a  local alignment algorithm used in biological morphology.} 
	\label{fig:ateles}
\end{figure*}   

\vspace{-10pt}
\paragraph{Evaluation of qBnB for rigid-bijective.}
We compare our qBnB algorithm for the rigid bijective problem with the BFS version of the BnB algorithm proposed in \cite{li20073d}. This algorithm modifies the standard $\ell^2$ energy we study here, and globally optimizes it over the same feasible set $\Gbi \times \Pibi$. in Figure~\ref{fig:hartley} we examined the performance of both algorithms on the problem of registering 2D frog silhouettes \cite{frogs} and 3D chair meshes \cite{giorgi2007shape}, sampled at $50$ points, with accuracy set at $\epsilon=10^{-6}$. Both algorithms returned identical results, which are visualized on the top of Figure~\ref{fig:hartley}. However qBnB took $0.5$ and $6.5$ seconds to solve the 2D and 3D problems respectively, while \cite{hartley2009global} required $11$ seconds to solve the 2D problem, and did not converge to the required accuracy for the 3D problem after $10^7$ evaluations (over $1$ hour).   The graphs in Figure~\ref{fig:hartley} show the number of evaluations the algorithms performs in each depth level of the BFS. In both cases qBnB required significantly less evaluations than \cite{li20073d}, and furthermore we note that asymptotically qBnB requires a constant number of evaluations at each depth, while \cite{li20073d} requires $\sim 2^{Dg/2}$ iterations at depth $g$, as shown by the log-scale aymptotes in the figure. This illustrates the theoretical results of Theorem~\ref{thm:complexity}: As for  BnB algorithms  $\epsilon \sim 2^{-g} $, they require $\sim \epsilon^{-D/2}  $ evaluations. In contrast, the dependence of qBnB on $g \sim \log (1/\epsilon) $ is linear.

\vspace{-10pt}
\paragraph{Application to biological morphology.}

%\subsection{Applications to morphology}

Biological morphology is concerned with the geometry of anatomical shapes. In this context, Auto3dgm \cite{boyer2015new,puente2013distances} is a popular algorithm that seeks for a plausible, consistent alignment of a collection of scanned morphological shapes. It is  based on solving the rigid bijective matching problem between pairs of shapes, and on a global step that synchronizes the rotations. The rigid bijective matching problem is solved using an alternating ICP-like algorithm, initialized from one of the $2^3$  rotations that align the principal axes of the two shapes. In this application, reflections  are typically allowed, and our algorithm can be easily extended to this case as $\Od$ is just a disjoint union of two copies of $\SO$.       

Figure~\ref{fig:ateles} compares the performance of qBnB to the pairwise matching algorithm of Auto3dgm for aligning a challenging collection of ten almost-symmetric teeth of spider monkeys (Ateles) from \cite{morph}. Details on the data for this experiment can be found in Appendix~\ref{app:morph}. Each bone surface is sampled at $200$ points using farthest point sampling. The pairwise energy between all pairs, and the alignment computed between the first tooth and all remaining teeth is shown for qBnB (top) and  Auto3dgm (bottom). We see that qBnB obtains lower energy solutions and more plausible alignments which have been verified as semantically correct by biological experts. This comes at the cost of higher complexity ($2-5$ minutes as opposed to $\sim 3$ seconds).

\vspace{-10pt}
\paragraph{Implementation details.}
The qBnB algorithm for rigid-CP was implemented in C++, based on the implementation of \cite{yang2016goicp}. The algorithm for rigid-bijective matching was implemented in Matlab; We solved for the bijective mapping used the excellent MEX implementation of \cite{bernard2016fast} for the auction algorithm \cite{bertsekas1998network}. Timings were measured using a Intel 3.10GHz CPU. 

%The frog silhouettes in Figure~\ref{fig:hartley} was created by Pedro de Siracusa and downloaded from PhyloPic, and the chair models come from \cite{giorgi2007shape}. The teeth models shown in Figure~\ref{fig:ateles} come from Morphosource, and additional details on each tooth appear in the supplementary material.  

%\sk{W? did I drop something here? also, maybe move to the end, or soften the blow a little?}with the exception of the timing of the multiple experiments in Figure~\ref{fig:goICP}(c) where the experiments were run in parallel on a 40-core CPU computer \sk{is the code itself parallel?},  timing presented in this section was measured on a standard PC with an Intel i5 3.30 GHz CPU. \sk{rephrase} The qBnB algorithm for rigid-CP was implemented in C++ \sk{say it's a modification of the goicp code...?} and the algorithm for rigid-bijective matching was implemented in Matlab. 

For the optimization of large scale rigid-CP with moderate accuracy $\epsilon=10^{-3} $ as in Figure~\ref{fig:goICP}(b) and Figure~\ref{fig:bunny}, we followed the implementation of GoICP \cite{yang2016goicp} and used the 3D Euclidean distance transform (DT) from \cite{fitzgibbon2003robust} with a $300 \times 300 \times 300$ grid; this provides a very fast approximation of the closest point computation, compared with the standard (accurate) KD-tree based computation. For evaluating the complexity as a function of accuracy in Figure~\ref{fig:goICP}(a), we used both GoICP and qBnB without DT transform.

%  since we cannot require high accuracy with the DT transform. In parts (b)-(c) we used the DT for both algorithms. To evaluate the error induced by the DT transform we compared the first ten experiments per noise level with the corresponding accurate solution obtained from our algorithm with DT computation and accuracy of $10^{-6}$. We find that  the DT transform introduces an error of up to $~0.005 $, slightly higher than the allowed error tolerance.

% The timing of both algorithms is roughly the number of CP evaluations multiplied by the complexity of each evaluation. In \cite{yang2016goicp} it is suggested to use the 3D Euclidean distance transform (DT) from \cite{fitzgibbon2003robust}  with a $300 \times 300 \times 300$ grid to achieve a very fast approximation of the CP computation in comparison with the standard (accurate) CP computation using KD trees. In Figure~\ref{fig:goICP}(a) we used both algorithms without the DT transform since we cannot require high accuracy with the DT transform. In parts (b)-(c) we used the DT for both algorithms. To evaluate the error induced by the DT transform we compared the first ten experiments per noise level with the corresponding accurate solution obtained from our algorithm with DT computation and accuracy of $10^{-6}$. We find that  the DT transform introduces an error of up to 

%\subsection{Comparison with Lipschitz method}

\section{Future work and conclusions}
 We presented the qBnB framework for globally optimizing $D$-quasi-optimizable functions, and demonstrated theoretically and empirically the advantage of this framework over existing BnB algorithm. Future challenges include applying the qBnB framework to other rigid registration problems that can handle outliers better than the standard $\ell^2$ energy, as well as using this framework for global optimization of $D$-quasi-optimizable functions in other knowledge domains.

\paragraph{Acknowledgements} The authors would like to thank Prof. Ingrid Daubechies for helpful discussions. Research is partially supported by Simons Math+X Investigators Award 400837.  
\newpage
{\small
\bibliographystyle{abbrv}
\bibliography{quasibib}
}
%\newpage

\newpage
\onecolumn
\appendix
\input{appendixInput}
\end{document}

%% file: appendixInput.tex
\section{Proofs}
\setcounter{equation}{0}
\numberwithin{equation}{section}
 We begin by proving that $D$-optimizable functions are Lipschitz on the cube $\Czero$ as we claimed in the beginning of Subsection~\ref{sub:optimizing}. In the remainder of the Section we restate and prove the theorems and lemmas stated in the main text.

\begin{lemma}
	If $F$ is $D$-quasi optimizable  in $\Czero$, then there exists $L>0$ such that.	\begin{equation}\label{eq:lip}
	F(x_1)-F(x_2) \leq L\|x_1-x_2\|, \forall x_1,x_2\in \Czero
	\end{equation}
\end{lemma} 
\begin{proof}
If  $F(x_1) \leq F(x_2)$ then \eqref{eq:lip} is obvious. Now assume $F(x_1)>F(x_2) $ and let $y_2$ be the minimizer of $E(x_2,\cdot) $. Then due to the differentiability of $E(\cdot,y_2) $,
$$F(x_1)-F(x_2) \leq E(x_1,y_2)-E(x_2,y_2)\leq L \|x_1-x_2\| $$
where $L$ is the maximum of the norm of the gradient of $E(\cdot,y) $ over all $x \in \Czero$ and $y \in \Y $. 
\end{proof}

\complexity*
\begin{proof}
\textbf{Part 1.} We begin with a general discussion of the complexity of Algorithm~\ref{alg} in both the BnB and quasi-BnB version, and prove the upper bound on $\nb$ and $\nq$ simultaneously. To do so we denote 
\begin{equation}\label{eq:DeltaAlpha}
\Delta_\alpha(\delta)=M \delta^\alpha, \alpha=1,2  .
\end{equation}
We denote by $n_\alpha$ the number of $F$-evaluations used by the algorithm in each case, so 
$$n_1=\nb, n_2=\nq .$$
We also note that while our bounds \eqref{eq:bijQlb},\eqref{eq:quasiLBicp} include besides $M \delta^2 $ also higher terms in $\delta$ these terms can be absorbed into $M\delta^2$ with a larger value of $M$. 

In the proof we will  call $g$ form Algorithm~\ref{alg}  the \emph{generation} of the algorithm. We begin with showing the algorithm terminates and bounding the final value of $g$ which we denote by $g_f$. First note if $\C_{h_i}(x_i) \in L_g $ contains a global minimum $x_*$, then its sub-cubes will always be added to $L_{g+1}$. This is because since  $\C_{h_i}(x_i)$ contains  a global minimum, $\lb_i $ is in fact a true lower bound for the minimal value of $F$ on the cube, which in this case is the global minimum $F_*$. Thus the value of $\lb_i\leq F_* \leq \ub$ and so the cube will be preserved.  It also follows that the global lower bound $\lb$ is in fact a true lower bound, since by definition it is smaller than the quasi-lower bound in a cube containing a global minimizer. It follows that if the algorithm does terminate, then the output $x_*$ is an $\epsilon$-optimal solution as its difference from the minimum $F^*$ satisfies
\begin{equation} 
F(x_*)-F^* = \ub-F^* \leq \ub-\lb \leq \epsilon  .
\end{equation}
Now to bound $g_f$ note that cubes in $L_g$ have half-edge length $h(g) \equiv h_0 2^{-g} $. the algorithm must terminate once it visits all cubes of generation $g$ whose edge length $h(g)$ satisfies
\begin{equation}\label{eq:deltaEpsilon} \Delta_\alpha(\sqrt D h(g)) \leq \epsilon .\end{equation} 
This is because for all cubes $\C_{h_j}(x_j)$ in this generation,
 $$\ub-\lb_j \leq \ub_j-\lb_j=\Delta_\alpha(\sqrt D h_j) \leq \epsilon .$$
 and by taking the minimum over $j$ we obtain that $\ub -\lb \leq \epsilon $. 
 
Some algebraic manipulation shows that \eqref{eq:deltaEpsilon} occurs when  
\begin{equation}\label{eq:gf}
g_f=g_f(\epsilon)=\lceil 1/\alpha \log \frac{\bar C_\alpha}{\epsilon} \rceil \text{ where } \bar C_\alpha=M(\sqrt D h_0)^\alpha  
\end{equation}
where we use $\log=\log_2$ throughout this proof. The number of $F$ evaluations $n_\alpha$ is bounded by  the worst case scenario where all cubes need to be divided in all generations 
\begin{align*}
n_\alpha\leq \sum_{g=0}^{g_f} 2^{Dg}=2^{Dg_f} \sum_{g=0}^{g_f} 2^{-Dg} \leq 2^{Dg_f} \sum_{g=0}^{\infty} 2^{-Dg}=2^{Dg_f}\frac{1}{1-2^D}\leq^\eqref{eq:gf} \left[\frac{2^D}{1-2^D} \bar C_\alpha^{D/\alpha} \right]\epsilon^{-D/\alpha}.
\end{align*} 
this proves the upper bound on $n_1,n_2 $. 

\textbf{Part 2.} We now show a lower bound on $\nb$. 
Our first step is to show that the sub-cubes of a given cube  have better=larger lower bounds. To see this let $\C^1=\C_{h_1}(x_1)$ be a cube, and  let $\C^2=\C_{h_2}(x_2) $ be one of its sub-cubes. Then $h_2=h_1/2 $ and $|x_1-x_2|=\sqrt{D}h_1/2 $. It follows that, using the notation of \eqref{eq:DeltaAlpha} with $\alpha=1$,
\begin{equation}\label{eq:x1x2}
F(x_2) \geq F(x_1)-M\sqrt{D}h_1/2
\end{equation}
Now denoting by $\lb_1$ and $\lb_2$ the lower bounds computed for the cubes $\C^1$ and $\C^2$ respectively, we have 
\begin{equation*}
\lb_2=F(x_2)-M\sqrt D h_1/2\geq^{\eqref{eq:x1x2}} F(x_1)-M\sqrt D h_1=\lb_1,
\end{equation*} 
and so we have $\lb_2 \ \geq \lb_1 $ as we stated.

Next note that by the quadratic bound \eqref{eq:quadratic} we have that for a global minimum $x_*$ and  $\eta=\sqrt{\epsilon/C} $,
\begin{equation} \label{eq:epsClose}
F(x)-F(x_*) \leq \epsilon, \text{ for all } x \in B_{\eta}(x_*) .
\end{equation}
Now, let $g_F$ denote the value of $g_f(2\epsilon) $ from \eqref{eq:gf}, for the case $\alpha=1$. Recall that $g_f(2\epsilon) $ is defined as the first integer for which \eqref{eq:deltaEpsilon} holds, where $\epsilon$ is replaced by $2\epsilon$, and $\alpha=1$. Thus for $g<g_F$ we have that 
\begin{equation}\label{eq:twoEps}
M \sqrt{D} h(g) >2\epsilon 
\end{equation}

 Let $\C_{h_i}(x_i)$ be a cube of generation $g_0$ containing $x_*$, where $g_0$ is large enough so that the diameter of the cube is smaller than $\eta$ and thus it is contained in $B_\eta(x_*)$. This occurs for 
$$g_0=\lceil \log(\bar C/\sqrt{\epsilon}) \rceil, \text{ where } \bar C=2h_0 \sqrt{CD}  .$$    

Every sub-cube $\C_{h_j}(x_j) $ of $\C_{h_i}(x_i)$, from any generation  $g_0 \leq g<g_F$, satisfies
\begin{equation}\label{eq:lbLow}
\lb_j=F(x_j)-M\sqrt{D}h_j <^{\eqref{eq:twoEps}} F(x_j)-2\epsilon\leq^{\eqref{eq:epsClose}} F(x_*)-\epsilon.
\end{equation}
In particular it follows that the cube $\C_{h_i}(x_i)$, and all its sub-cubes, will be visited during the BnB search. This is because we saw that lower bounds improve by refinement, and so by \eqref{eq:lbLow} any cube from the earlier generations $g<g_0$ which contains $\C_{h_i}(x_i)$, also has lower bounds which are lower that the global minimum (by at least $\epsilon$) and so such a cube would not be removed from the search.. 

We can now bound $\nb$ by the number of subcubes of $\C_{h_i}(x_i) $ at the $g_F-1$ generation alone: 
\begin{equation}\label{eq:nbLB}
\nb \geq 2^{D(g_F-1-g_0)}
\end{equation}
Now 
$$g_F-1-g_0=-1+\lceil \log \frac{\bar C_1}{2\epsilon} \rceil-\lceil \log(\bar C/\sqrt{\epsilon}) \rceil \geq \log(\bar C_1/(2\bar C \sqrt{\epsilon}))-2=\log(\bar C_1/(8\bar C \sqrt{\epsilon}))  .$$
So returning to \eqref{eq:nbLB} we obtain 
$$\nb \geq 2^{D(g_f-1-g_0)} \geq \left(\frac{\bar C_1}{8\bar C} \right)^D \epsilon^{-D/2}  .$$

\textbf{Part 3.} We now turn to prove the last part of the theorem. Let $\J(\ell)$ denote the set of indices $k$ for which $(x_{\ell}^*,y_k^*) $ is a minimizer. Note the we always have that $\ell$ is in $\J_{\ell}$. Let $m$ be half of the minimun over the minimal eigenvalue of the hessian of $E(\cdot,y_{\ell}^*) $ at $x_{\ell}^* $ for all $\ell$.  The assumption that $E$ has a finite number $N$ of minimizers $x_{\ell}^*,y_{\ell}^* $, with strictly positive definite hessian, implies that $m>0$, and so for small enough positive $\eta$,
\begin{equation} \label{eq:quadBelow}
F(x)-F(x_\ell^*)=\min_{k \in \J(\ell)} E(x,y_k^*)-E(x_\ell^*,y_k^*) \geq m\|x-x_\ell^*\|^2, \, \forall 1 \leq \ell \leq N \text{ and } \forall x \in B_{\eta}(x_\ell^*)  .
\end{equation}

The minimum of $F$ on $\Czero\setminus \bigcup_i B_{\eta/2}(x_\ell^*) $ is strictly larger than $F^*$. Therefore there exists some $g_0 $ independent of $\epsilon$, such that all cubes of generation $g_0$ which are not contained in one of the balls $B_{\eta}(x_\ell^*) $ will be removed in the $g_0$-th stage.

We now claim that for $g \geq g_0$, $g$-th generation cubes $\C_{h_i}(x_i) $ contained in one of the  balls $B_{\eta}(x_\ell^*) $   will be removed if 
\begin{equation}\label{eq:necessary}
\|x_i-x_\ell^*\|_{\infty} > \sqrt{ \frac{2MD}{m} } h_i=\sqrt{\frac{2\Delta_*(\sqrt{D}h_i)}{m}}.
\end{equation}
This is because
\begin{align*}
\lb_j&=F(x_i)-\Delta_*(\sqrt{D}h_i) \geq^{\eqref{eq:quadBelow}} F(x_\ell^*)+m\|x_i-x_\ell^*\|^2-\Delta_*(\sqrt{D}h_i)\\
&=\ub+(F(x_\ell^*)-\ub)+m\|x_i-x_\ell^*\|^2-\Delta_*(\sqrt{D}h_i)\geq^{(*)} \ub+m\|x_i-x_\ell^*\|^2-2\Delta_*(\sqrt{D}h_i)\\
&\geq \ub+m\|x_i-x_\ell^*\|_{\infty}^2-2\Delta_*(\sqrt{D}h_i)  > ^{\eqref{eq:necessary}}\ub 
\end{align*}
where $(*)$ follows from the fact that if $\C_{h_i}(x_i) $ is the $g$-th generation cube containing $x_\ell^* $, then
$$F(x_\ell^*)-\ub \geq F(x_\ell^*)-\ub_i=F(x_\ell^*)-F(x_i) \geq -\Delta_*(\sqrt{D}h_i) .$$

Now for $g \geq g_0$, the condition \eqref{eq:necessary} is not fulfilled in at most $\bar C=(\sqrt{ \frac{2MD}{m} }+2)^D  $ cubes surrounding each minimizer, and so in total only $N \bar C $ cubes can survive each generation $g >g_0 $. The important point is that this number is independent of $\epsilon$. So the total number of iterations is bounded by the sum of the total number of cubes in all generations  $g \leq g_0$, which is some constant indepentent of $\epsilon$ which we denote by $b$, and the constant $N\bar C$ multiplied by the remaining number of iterations $g_f-g_0 $, that is  
$$n_2 \leq b+(g_f-g_0)N \bar C \leq n_2 \leq b+g_fN \bar C \leq^{\eqref{eq:gf}} b+ N \bar C (1/2 \log \frac{C_2}{\epsilon} +1)  $$   
This bound can be replaced with a bound of the form \eqref{eq:linearConvergence} with an appropriate constant. 
\end{proof}

\qlbBij*
\begin{proof}
To conclude the proof of the theorem for the case $r_*=0$ we need to show
\begin{lemma}\label{sub:norms}
	For all $r \in \RR^D$,
	\begin{equation}\label{eq:operator}
	\|[r]\|_{\mathrm{op}} \leq \|r\|
	\end{equation}
\end{lemma}
\begin{proof}
	The non-zero eigenvalues of a skew-symmetric real matric $[r] $ can be written as
	$$a_1 i, -a_1i, a_2i,-a_2i,\ldots $$
	where $a_1 \geq a_{2}\ldots >0 $. Therefore
	$$\|[r]\|_{\mathrm{op}}^2=a_1^2 \leq 1/2\sum_ia_i^2=1/2\|[r]\|_F^2=\|r\|^2 .$$	
\end{proof}

	For the general case $r_*\neq 0 $, we define a change of variable $\tilde p_i=R_{r_*}p_i $ and denote by $\EbiTilde$ the energy resulting by replacing $p_i$ with $\tilde p_i$ in the definition of $\Ebi$. Then for all $R_0,\pi $ we have
$$\EbiTilde(R_0R_{r_*}^T,\pi)=\Ebi(R_0,\pi). $$
In particular $\tilde{r}_*=0$ is a minimizer of $\FbiTilde$ which is defined by replacing $\Ebi$ with $\EbiTilde$ in the definition of $\Fbi$. We claim that there exists $r_1$ such that 
\begin{equation} \label{eq:comm}
R_{r_1}=R_rR_{r_*}^T \text{ and } \|r_1\| \leq \|r-r_*\|
\end{equation}
In the case $d=2 $ we can identify $R_rR_{r_*}^T $ with $e^{i(r-r_*)} $ and so we can simply choose $r_1=r-r_* $. For $d=3$ it is proven in Lemma~3.2 in \cite{hartley2009global} that the angular distance between $R_r $ and $R_{r_*} $ is smaller or equal to $\|r-r_*\| $. This means that the angular distance between $R_rR_{r_*}^T$ and the identity is less than $\|r-r_*\|$ and thus implies the existence of $r_1$ satisfying \eqref{eq:comm}. 

Now using the bound from \ref{eq:bijQlb} for $\tilde{F}$ which is minimized at zero we obtain
$$\Fbi(r)-\Fbi(r_*)=\FbiTilde(r_1)-\FbiTilde(0)\leq \frac{2}{n}\psi_2(\delta)\sigma_{\tilde \P}\sigma_{\Q}=\frac{2}{n}\psi_2(\delta)\sigma_{\P}\sigma_{\Q} \quad  .$$

\end{proof}

\qlbCP*

\begin{proof}
	The proof is very similar to the proof of Theorem~\ref{thm:quasiBi}. Let us first consider the case $(r_*,t_*)=(0,0) $, and let $\pi_*$ be the corresponding mapping so that $(\Id,0,\pi_*) $ minimizes $\ECP$. Then 
	\begin{align*}
	\FCP(r,t)-\FCP(0,0)&\leq \ECP(r,t,\pi_*)-\ECP(0,0,\pi_*)\\
	&=\frac{1}{n}\sum_{i=1}^n\left[2\langle (\Id-R_r)p_i,q_{\pi_*(i)} \rangle+2\langle R_rp_i,t\rangle-2\langle t,q_{\pi_*(i)} \rangle+\|t\|^2 \right]\\
	&=^{(*)}\frac{1}{n}\sum_{i=1}^n \left[2\sum_{k=2}^\infty \frac{1}{k!}\langle[r]^kp_i,q_{\pi_*(i)}\rangle+2\sum_{k=1}^\infty \frac{1}{k!}\langle[r]^kp_i,t\rangle]+\|t\|^2 \right]\\
	&\leq \frac{1}{n}\sum_{i=1}^n \left[2\sum_{k=2}^{\infty}\frac{1}{k!}\|r\|_{\mathrm{op}}^k\|p_i\|(\|p_i\|+\|q_{\pi_*(i)}-p_i\|)+2\sum_{k=1}^{\infty}\frac{1}{k!}\|r\|_{\mathrm{op}}^k\|p_i\| \ \|t\|+\|t\|^2 \right]\\
	&\leq  \frac{1}{n} \left[2\psi_2(\delta_1) (\sigma_{\P}^2+\sigma_{\P}[\sum_i \|q_{\pi_*(i)}-p_i\|^2]^{1/2})  +2 \psi_1(\delta_1)\delta_2 \sum_i \|p_i\|+n\delta_2^2  \right]\\
	&\leq  \frac{1}{n} \left[2\psi_2(\delta_1) (\sigma_{\P}^2+\sigma_{\P}\sqrt{n f_*})  +2 \psi_1(\delta_1)\delta_2 \sum_i \|p_i\|+n\delta_2^2  \right]
	\end{align*}
	Here $(*)$ follows from the fact that $\ECP(\cdot,\cdot,\pi_*) $ is minimized at the origin and so the first order terms cancel out, and the next inequalities follow from the Cauchy-Schwarz inequality and from Lemma~\ref{sub:norms}.
	
	For general $(r_*,t_*) $, we use a change of variables $\tilde p_i=R_*p_i, \, \tilde q_i= q_i-t_* $, and denote by $\ECPtilde$ and $\FCPtilde$ the functions obtained by replacing $p_i,q_i$ by $\tilde p_i,\tilde q_i $ in the definition of these functions. For given $(r,t,\pi) $ we have
	$$\ECP(R_r,t,\pi)=\ECPtilde(R_rR_*^T,t-t_*,\pi) $$
	Let $r_1\in \RR^D$ satisfying \eqref{eq:comm}. By applying the theorem to $\FCPtilde$ which is minimized at $(0,0)$ we obtain
	\begin{align*}
	\FCP(r,t)-\FCP(0,0)&=\FCPtilde(r_1,t-t_*)-\FCPtilde(0,0) \\
	&\leq \frac{1}{n} \left[2\psi_2(\delta_1) (\sigma_{\P}^2+\sigma_{\P}\sqrt{n f_*})  +2 \psi_1(\delta_1)\delta_2 \sum_i \|p_i\|+n\delta_2^2  \right] 
	\end{align*}
\end{proof}

\section{BnB for rigid closest point}
In the following we explain how we construct a quasi-BnB framework for the rigid CP problem based on the BnB architecture proposed in Go-ICP \cite{yang2016goicp}. Following Go-ICP, we use a nested BnB structure: We perform an ``outer'' BnB search on the rotation space, wherein the upper and lower bounds are functions of the translation component $t$; in turn, to compute these bounds we perform an ``inner'' BnB search over the variable $t$. Namely, for given $r_i$, we define an upper bound for the outer BnB by
\begin{equation}
\ECPbar(r_i)=\min_{t \in \C_1(0),\pi \in \PiCP} \ECP(r_i,t,\pi).
\end{equation}

To compute a quasi-lower bound for the outer BnB we note that if $(r_*,t_*)$ minimizes $\FCP$ and $r_* \in \C_h(r_i)$, then by using  \eqref{eq:quasiLBicp} where we set $r=r_i $ , take $\delta_1 $ to be the maximal distance of a point in the cube from the center, $t=t_*$ and  $\delta_2=0$  we obtain
\begin{equation}
\FCP(r_*,t_*)\geq \FCP(r_i,t_*)-\frac{2}{n}\left(1+\sqrt{\frac{f_*}{\sigma_{\P}}}\right)\sigma_{\P}\psi_2(\sqrt{D}h),
\end{equation}
and since $\FCP(r_i,t_*)\geq \ECPbar(r_i)$ it follows that if $r_* \in \C_h(r_i)$ then
\begin{equation}
\FCP(r_*,t_*) \geq \ECPbar(r_i)-\frac{2}{n}\left(1+\sqrt{\frac{f_*}{\sigma_{\P}}}\right)\sigma_{\P}\psi_2(\sqrt{D}h).
\end{equation}
The RHS of the equation above gives us our quasi lower bound for the rotation quasi BnB. To compute $\ECPbar(r_i) $ we compute a BnB in translation space, where throughout the translation BnB the rotation coordinate $r_i$ is fixed. For a given translation cube $\C_h(t_j)$ an upper bound for the value of $\ECPbar(r_i) $ is given by evaluation of $\FCP(r_i,t_j)$. If $t_*$ is a minimizer of $\FCP(r_i,\cdot) $ then a quasi-lower bound in the cube is given by  
\begin{equation}
\ECP(r_i,t_*)\geq \ECP(r_i,t_j)-\frac{dh^2}{n}.
\end{equation} 
We note that this bound is similar to what we would get by setting $\delta_1=0$ and $\delta_2$ to be the maximal distance in the cube from $t_j$ in \eqref{eq:quasiLBicp}. Although the bound does not follow directly from this equation the derivation is similar, and can be obtained by studying the behavior of a minimizer of $E(r_i,\cdot,\pi_*) $, so we do not go into the details. Finally we note that when the quasi-lower bounds in the outer or inner BnB is lower than zero we replace it with zero.

The rest of the architecture of the BnB is also borrowed from Go-ICP. We use best-first-search, where the cube with the lowest lower bound is visited first. Every time the upper bound is improved, an ICP algorithm is run to improve the resolution of the solution. For more details see \cite{yang2016goicp}.

\section{Morphological data}\label{app:morph}
The morphological data for the experiment shown in Figure~\ref{fig:ateles} comes from the MorphoSource dataset. The figure shows ten different second mandibular molars of spider monkeys (Ateles), which come from three different taxonomical groups. More details on the data are shown in the table below. 

\begin{figure*}[h]
	\includegraphics[width=\textwidth]{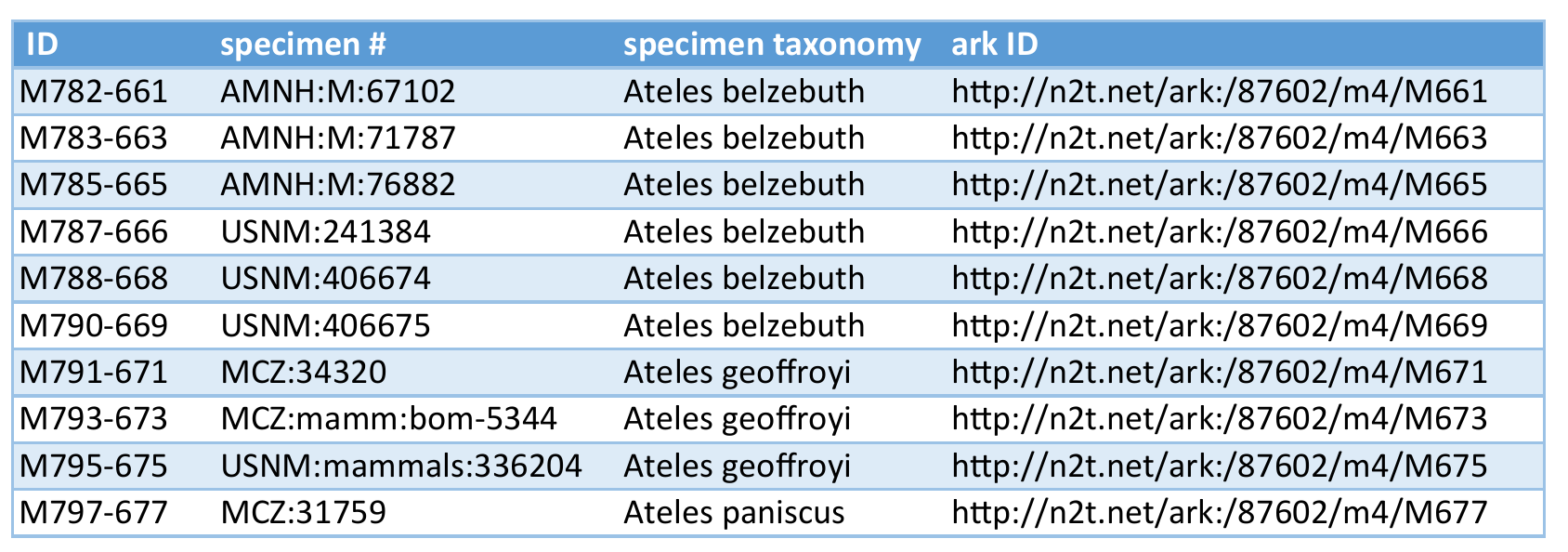}
	\vspace{-15pt}
%	\caption{} 
%	\label{fig:ateles}
\end{figure*}    

%% file: paper.bbl
\begin{thebibliography}{10}

\bibitem{morph}
The morphosource database, \url{https://www.morphosource.org/}.

\bibitem{frogs}
The pylopic database, \url{http://phylopic.org/}.

\bibitem{aiger20084}
D.~Aiger, N.~J. Mitra, and D.~Cohen-Or.
\newblock 4-points congruent sets for robust pairwise surface registration.
\newblock In {\em ACM Transactions on Graphics (TOG)}, volume~27, page~85. Acm,
  2008.

\bibitem{al2013continuous}
R.~Al-Aifari, I.~Daubechies, and Y.~Lipman.
\newblock Continuous procrustes distance between two surfaces.
\newblock {\em Communications on Pure and Applied Mathematics}, 66(6):934--964,
  2013.

\bibitem{bernard2016fast}
F.~Bernard, N.~Vlassis, P.~Gemmar, A.~Husch, J.~Thunberg, J.~Goncalves, and
  F.~Hertel.
\newblock Fast correspondences for statistical shape models of brain
  structures.
\newblock In {\em Medical Imaging 2016: Image Processing}, volume 9784, page
  97840R. International Society for Optics and Photonics, 2016.

\bibitem{bertsekas1998network}
D.~P. Bertsekas.
\newblock {\em Network optimization: continuous and discrete models}.
\newblock Athena Scientific Belmont, 1998.

\bibitem{besl1992method}
P.~J. Besl and N.~D. McKay.
\newblock Method for registration of 3-d shapes.
\newblock In {\em Sensor Fusion IV: Control Paradigms and Data Structures},
  volume 1611, pages 586--607. International Society for Optics and Photonics,
  1992.

\bibitem{boyd2004convex}
S.~Boyd and L.~Vandenberghe.
\newblock {\em Convex optimization}.
\newblock Cambridge university press, 2004.

\bibitem{boyer2015new}
D.~M. Boyer, J.~Puente, J.~T. Gladman, C.~Glynn, S.~Mukherjee, G.~S. Yapuncich,
  and I.~Daubechies.
\newblock A new fully automated approach for aligning and comparing shapes.
\newblock {\em The Anatomical Record}, 298(1):249--276, 2015.

\bibitem{breuel2003implementation}
T.~M. Breuel.
\newblock Implementation techniques for geometric branch-and-bound matching
  methods.
\newblock {\em Computer Vision and Image Understanding}, 90(3):258--294, 2003.

\bibitem{bustos2016fast}
{\'A}.~P. Bustos, T.-J. Chin, A.~Eriksson, H.~Li, and D.~Suter.
\newblock Fast rotation search with stereographic projections for 3d
  registration.
\newblock {\em IEEE Transactions on Pattern Analysis and Machine Intelligence},
  38(11):2227--2240, 2016.

\bibitem{cygan2015parameterized}
M.~Cygan, F.~V. Fomin, {\L}.~Kowalik, D.~Lokshtanov, D.~Marx, M.~Pilipczuk,
  M.~Pilipczuk, and S.~Saurabh.
\newblock {\em Parameterized algorithms}, volume~4.
\newblock Springer, 2015.

\bibitem{dym2017exact}
N.~Dym and Y.~Lipman.
\newblock Exact recovery with symmetries for procrustes matching.
\newblock {\em SIAM Journal on Optimization}, 27(3):1513--1530, 2017.

\bibitem{fitzgibbon2003robust}
A.~W. Fitzgibbon.
\newblock Robust registration of 2d and 3d point sets.
\newblock {\em Image and vision computing}, 21(13-14):1145--1153, 2003.

\bibitem{fowkes2013branch}
J.~M. Fowkes, N.~I. Gould, and C.~L. Farmer.
\newblock A branch and bound algorithm for the global optimization of hessian
  lipschitz continuous functions.
\newblock {\em Journal of Global Optimization}, 56(4):1791--1815, 2013.

\bibitem{gelfand2005robust}
N.~Gelfand, N.~J. Mitra, L.~J. Guibas, and H.~Pottmann.
\newblock Robust global registration.
\newblock In {\em Symposium on geometry processing}, volume~2, page~5. Vienna,
  Austria, 2005.

\bibitem{giorgi2007shape}
D.~Giorgi, S.~Biasotti, and L.~Paraboschi.
\newblock Shape retrieval contest 2007: Watertight models track.
\newblock {\em SHREC competition}, 8(7), 2007.

\bibitem{hartley2009global}
R.~I. Hartley and F.~Kahl.
\newblock Global optimization through rotation space search.
\newblock {\em International Journal of Computer Vision}, 82(1):64--79, 2009.

\bibitem{irani1999combinatorial}
S.~Irani and P.~Raghavan.
\newblock Combinatorial and experimental results for randomized point matching
  algorithms.
\newblock {\em Computational Geometry}, 12(1-2):17--31, 1999.

\bibitem{khoo2016non}
Y.~Khoo and A.~Kapoor.
\newblock Non-iterative rigid 2d/3d point-set registration using semidefinite
  programming.
\newblock {\em IEEE Transactions on Image Processing}, 25(7):2956--2970, 2016.

\bibitem{li20073d}
H.~Li and R.~Hartley.
\newblock The 3d-3d registration problem revisited.
\newblock In {\em 2007 IEEE 11th International Conference on Computer Vision},
  pages 1--8. IEEE, 2007.

\bibitem{maron2016point}
H.~Maron, N.~Dym, I.~Kezurer, S.~Kovalsky, and Y.~Lipman.
\newblock Point registration via efficient convex relaxation.
\newblock {\em ACM Transactions on Graphics (TOG)}, 35(4):73, 2016.

\bibitem{mount1999efficient}
D.~M. Mount, N.~S. Netanyahu, and J.~Le~Moigne.
\newblock Efficient algorithms for robust feature matching.
\newblock {\em Pattern recognition}, 32(1):17--38, 1999.

\bibitem{parra2014fast}
A.~Parra~Bustos, T.-J. Chin, and D.~Suter.
\newblock Fast rotation search with stereographic projections for 3d
  registration.
\newblock In {\em Proceedings of the IEEE Conference on Computer Vision and
  Pattern Recognition}, pages 3930--3937, 2014.

\bibitem{pfeuffer2012discrete}
F.~Pfeuffer, M.~Stiglmayr, and K.~Klamroth.
\newblock Discrete and geometric branch and bound algorithms for medical image
  registration.
\newblock {\em Annals of Operations Research}, 196(1):737--765, 2012.

\bibitem{pottmann2006geometry}
H.~Pottmann, Q.-X. Huang, Y.-L. Yang, and S.-M. Hu.
\newblock Geometry and convergence analysis of algorithms for registration of
  3d shapes.
\newblock {\em International Journal of Computer Vision}, 67(3):277--296, 2006.

\bibitem{puente2013distances}
J.~Puente.
\newblock Distances and algorithms to compare sets of shapes for automated
  biological morphometrics.
\newblock 2013.

\bibitem{rangarajan1997softassign}
A.~Rangarajan, H.~Chui, and F.~L. Bookstein.
\newblock The softassign procrustes matching algorithm.
\newblock In {\em Biennial International Conference on Information Processing
  in Medical Imaging}, pages 29--42. Springer, 1997.

\bibitem{sandhu2010point}
R.~Sandhu, S.~Dambreville, and A.~Tannenbaum.
\newblock Point set registration via particle filtering and stochastic
  dynamics.
\newblock {\em IEEE transactions on pattern analysis and machine intelligence},
  32(8):1459--1473, 2010.

\bibitem{tam2013registration}
G.~K. Tam, Z.-Q. Cheng, Y.-K. Lai, F.~C. Langbein, Y.~Liu, D.~Marshall, R.~R.
  Martin, X.-F. Sun, and P.~L. Rosin.
\newblock Registration of 3d point clouds and meshes: a survey from rigid to
  nonrigid.
\newblock {\em IEEE transactions on visualization and computer graphics},
  19(7):1199--1217, 2013.

\bibitem{standfordBunny}
G.~Turk and M.~Levoy.
\newblock The stanford 3d scanning repository.
  \url{http://graphics.stanford.edu/data/3Dscanrep/}.

\bibitem{wachowiak2004approach}
M.~P. Wachowiak, R.~Smol{\'\i}kov{\'a}, Y.~Zheng, J.~M. Zurada, and A.~S.
  Elmaghraby.
\newblock An approach to multimodal biomedical image registration utilizing
  particle swarm optimization.
\newblock {\em IEEE Transactions on evolutionary computation}, 8(3):289--301,
  2004.

\bibitem{yang2016goicp}
J.~Yang, H.~Li, D.~Campbell, and Y.~Jia.
\newblock Go-icp: A globally optimal solution to 3d icp point-set registration.
\newblock {\em IEEE Transactions on Pattern Analysis and Machine Intelligence
  (T-PAMI)}, 38(11):2241--2254, 2016.

\bibitem{yang2013goicp}
J.~Yang, H.~Li, and Y.~Jia.
\newblock Go-icp: Solving 3d registration efficiently and globally optimally.
\newblock In {\em Proceedings of the 14th International Conference on Computer
  Vision (ICCV)}, pages 1457--1464, 2013.

\bibitem{zhou2016fast}
Q.-Y. Zhou, J.~Park, and V.~Koltun.
\newblock Fast global registration.
\newblock In {\em European Conference on Computer Vision}, pages 766--782.
  Springer, 2016.

\end{thebibliography}
